\documentclass[onecolumn]{autart}    

\usepackage{color}
\usepackage{amsmath,amssymb,amsfonts}
\usepackage[bb=boondox]{mathalfa}
\usepackage{stmaryrd}
\usepackage{graphicx}          

\usepackage{cite}
\usepackage{enumitem}

\newcommand{\bR} { {\mathbb R}}
\newcommand{\bZ} { {\mathbb Z}}

\newcommand{\cK} { {\mathcal K}}
\newcommand{\cU} { {\mathcal U}}
\newcommand{\cX} { {\mathcal X}}

\newcommand{\sD} {\mathsf{D}}

\newcommand{\ux} { {\cU \to \cX} }

\newcommand{\0}{\mathbb{0}}
\newcommand{\I}{\mathbb{I}}

\usepackage{xcolor}

\definecolor{fblightblue}{RGB}{220,235,250}

\definecolor{yklightred}{RGB}{250,235,220}

\def\red{\hfill $\lhd$}

\allowdisplaybreaks[4]

\begin{document}

\begin{frontmatter}

\title{Input-to-State Stability Implications in\\ Contraction Theory\thanksref{footnoteinfo}} 

\thanks[footnoteinfo]{
This paper was not presented at any IFAC 
meeting. Corresponding author Y.~Kawano. Tel. +81-82-424-4210.
}

\author[JP]{Yu Kawano}\ead{ykawano@hiroshima-u.ac.jp},    
\author[USA]{Francesco Bullo}\ead{bullo@ucsb.edu},               

\address[JP]{Graduate School of Advanced Science and Engineering, Hiroshima University, Higashi-Hiroshima 739-8527, Japan}  
\address[USA]{Department of Mechanical Engineering and the Center for Control, Dynamical Systems, and Computation, University of California, Santa Barbara, USA}             

\begin{keyword}                           
Contraction Theory; Incremental Stability; Input-to-State Stability.  
\end{keyword}                             

\begin{abstract}                          

    For nonlinear control systems on normed vector spaces, we
      characterize an incremental input-to-state stability (ISS) type property in
      which the overshoot constant multiplies both the initial-condition
      and the input terms. Working through the associated variational
      system, we show that two properties are equivalent: an ISS-type bound on the variational system, and the
      incremental ISS-type bound on the original system. We further establish the equivalence between an infinitesimal contraction condition, expressed through a Lyapunov-type function, and an incremental Lyapunov condition. Each of these equivalent conditions yields a necessary condition and a sufficient condition for the ISS-type bounds, differing only in the input Lipschitz constant of the vector field. When the overshoot constant equals one, the infinitesimal contraction condition reduces to the standard norm-based
      contraction condition. We establish these implications under mere
      continuous differentiability of the vector field, and we illustrate
      the results through sensitivity matrices and Lyapunov characteristic
      exponents. Moreover, we develop similar implications for discrete-time systems.

\end{abstract}

\end{frontmatter}

\section{Introduction}

Contraction theory~\cite{LW-SJJE:98,FF-RS:14,FB:26} provides a framework to study stability between the pair of trajectories of a system, known as \emph{incremental stability}. Often, incremental stability is analyzed by using the associated variational system. For control systems, a central role is played by \emph{incremental input-to-state stability} (ISS), also referred to as $\delta$ISS. Contractivity and incremental ISS are closely related, and it is a classic problem to characterize the equivalence between these two properties.

The study of this equivalence builds on the influential works~\cite{DA:02,DA:09} and~\cite{VA-BJ-LP:16}, along with several recent contributions. Most references address the Riemannian or, more generally, Finslerian setting~\cite{LW-SJJE:98,FF-RS:14,PG:15,VA-BJ-LP:16,YK-BB-MC:20,MG-DA-VA:23,YK-BB:24} for continuous-time systems; the normed vector space setting is considered in~\cite{NB-RO-AP:19,AD-SJ-FB:22}. In the classical Lyapunov setting, necessary and sufficient characterizations of incremental ISS are established in~\cite{DA:02,DA:09}.

The contributions of this paper are as follows. First, as the main contribution, we establish a necessary condition and a sufficient condition for an incremental ISS-type property in the contractivity framework through the analysis of variational systems. These conditions differ only in the input Lipschitz constant of the vector field. To the best of the authors' knowledge, no such characterization has been previously available. We further provide an equivalent incremental Lyapunov characterization. Second, unlike the aforementioned Finslerian setting, our analysis relaxes twice continuous differentiability of the vector field to continuous differentiability by showing that the transition matrix of the variational system coincides with the G\^ateaux derivative of the flow of the original system. A more detailed comparison with the literature can be found in Section~\ref{sec:ca}. Also, we illustrate some implications of our results for sensitivity functions~\cite{MZ-LX-GFT:23} and Lyapunov characteristic exponents~\cite{SW:03}. Finally, we develop similar ISS-type implications for discrete-time systems.


{\it Notation:}
The sets of real numbers and integers are denoted by~$\bR$ and~$\bZ$, respectively. The sets of real numbers and integers that are not less than ~$a \in \bR$ are denoted by~$\bR_{\ge a}$ and~$\bZ_{\ge a}$, respectively. Similarly, the sets of real numbers and integers that are greater than~$a \in \bR$ are denoted by~$\bR_{> a}$ and~$\bZ_{> a}$, respectively.
The~$n$-component vector and~$n \times m$-matrix whose all components are~$0$ are denoted by~$\0_n$ and~$\0_{n \times m}$, respectively. The~$n \times n$ identity matrix is denoted by~$\I_n$. For differentiable~$f:\bR \times \bR^n \times \bR^m \to \bR^n$, the partial derivatives of~$f(t,x,u)$ with respect to~$x$ and~$u$ are denoted by~$\sD_x f(t,x,u)$ and~$\sD_u f(t,x,u)$, respectively. 
Given vector norms~$\| \cdot \|_{\cX}$ and~$\| \cdot \|_\cU$, the induced matrix norm of~$\sD_u f(t,x,u)$ is denoted by
\begin{align*}
    \| \sD_u f \|_{\ux} := \sup_{(t,x,u)\in \bR \times \bR^n \times \bR^m} \| \sD_u f (t,x,u) \|_{\ux}, \quad 
    \| \sD_u f (t,x,u) \|_{\ux} := \max_{\substack{\|v\|_\cU = 1}} \|\sD_u f(t,x,u) v\|_{\cX}.    
\end{align*}
This is equivalent to the Lipschitz constant of~$f(t,x,u)$ in~$u$ from~$\| \cdot \|_\cU$ to~$\| \cdot \|_{\cX}$.
Where~$\| \cdot \|_\cU =\| \cdot \|_\cX$, we simply write~$\| \cdot \|_\cX$, i.e.,
\begin{align*}
    \| \sD_x f \|_\cX := \sup_{(t,x,u)\in \bR \times \bR^n \times \bR^m} \| \sD_x f (t,x,u) \|_\cX , \quad \
    \| \sD_x f (t,x,u) \|_\cX :=\max_{\substack{\|v\|_\cX = 1}} \|\sD_x f(t,x,u) v\|_{\cX}.
\end{align*}



\section{Input-to-State Stability Equivalences}

\subsection{Variational Systems}

Consider a nonlinear control system:
\begin{align}\label{eq:sys}
	\dot x = f(t, x, u),
\end{align}
where~$f: \bR \times \bR^n \times \bR^m \to \bR^n$. Let~$\phi (t, t_0, x_0, u)$ denote the solution to the system~\eqref{eq:sys} at time~$t \in \bR$ with initial condition~$x(t_0) = x_0 \in \bR^n$ and input~$u:\bR \to \bR^m$. Namely,~$x(t) = \phi (t, t_0, x_0, u)$. We impose the following assumption for~$f(t,x,u)$.

\begin{assum}[Continuous Differentiability of Vector Field]\label{asm:vf}
  The vector field~$f(t,x,u)$ and its Jacobians $\sD_x f(t,x,u)$ and~$\sD_u f(t,x,u)$ exist and are jointly continuous in~$(t,x,u) \in \bR \times \bR^n \times \bR^m$. 
    \red
\end{assum}

\begin{rem}[Local Unique Existence of Solution]\label{rem:sol}
    By the Picard--Lindel\"of theorem~\cite[Theorem 1.14]{FB:26}, Assumption~\ref{asm:vf} implies that for each initial condition~$(t_0, x_0) \in \bR \times \bR^n$ and every continuous input~$u:\bR \to \bR^m$, there exists~$t_1 = t_1 (t_0, x_0, u)\in \bR_{>t_0}$ such that the solution~$\phi (t, t_0, x_0, u)$ to the system~\eqref{eq:sys} exist and are unique and continuously differentiable with respect to~$t \in [t_0, t_1]$. Moreover, by continuous dependence on initial conditions, for each continuous~$u:\bR \to \bR^m$, its partial derivative~$\sD_{x_0} \phi^u (t, t_0, x_0)$ with~$\phi^u (t, t_0, x_0) := \phi (t, t_0, x_0, u)$ exists and is jointly continuous in~$(t, t_0, x_0)$ on the domain of existence.
    \red
\end{rem}

Our analysis relies on a fundamental relation between the solutions to the system and its variational system along~$\phi (t, t_0, x_0, u)$ with the initial condition~$\delta x(t_0) = \delta x_0 \in \bR^n$ and input~$\delta u:\bR \to \bR^m$:
\begin{align}\label{eq:vsys}
	\dot {\delta x} = \bigl(\sD_x f(t, x_t, u(t)) \delta x + \sD_u f(t, x_t, u(t)) \delta u\bigr)\bigr|_{x_t=\phi (t, t_0, x_0, u)}.
\end{align}

\begin{thm}[G\^ateaux Derivative of Solution]\label{thm:GD}
	Under Assumption~\ref{asm:vf}, for each initial condition~$(t_0, x_0) \in \bR \times \bR^n$ and every continuous input~$u:\bR \to \bR^m$, there exists~$t_1 = t_1 (t_0, x_0, u)\in \bR_{>t_0}$ such that the solution~$\delta x(t)$ to the variational system~\eqref{eq:vsys} satisfies
    \begin{align}\label{eq:sol_GD}
        \delta x(t) = \lim_{h \to 0^+} \frac{\phi (t, t_0, x_0 + h \delta x_0, u + h \delta u) - \phi (t, t_0, x_0, u)}{h}
	\end{align}
    for each~$t \in [t_0, t_1]$, every initial condition~$\delta x(t_0)=\delta x_0 \in \bR^n$, and every continuous input~$\delta u:\bR \to \bR^m$.
\end{thm}

\begin{pf}
    The proof is shown in Section~\ref{pf:GD}. \qed
\end{pf}

Since $\phi (t, t_0, x_0, u)$ is continuously differentiable in~$x_0$ by Assumption~\ref{asm:vf}, \eqref{eq:sol_GD} implies
\begin{align}\label{eq:tm}
    \delta x(t)
    = \lim_{h \to 0^+} \frac{\phi (t, t_0, x_0 + h \delta x_0, u ) - \phi (t, t_0, x_0, u)}{h}
    = \sD_{x_0} \phi (t, t_0, x_0, u) \delta x_0
\end{align}
for each~$t \in [t_0, t_1]$ and every initial condition~$\delta x(t_0)=\delta x_0 \in \bR^n$, when~$\delta u(t) \equiv \0_m$. This means that~$\sD_{x_0} \phi (t, t_0, x_0, u)$ is the transition matrix of the variational system~\eqref{eq:vsys} as a linear time-varying system. Namely, we have
\begin{subequations}
		\begin{align}
			\sD_{\tau} \sD_{x_0} \phi (\tau, t_0, x_0, u)
            &= \sD_x f(\tau, x(\tau), u(\tau)) 
            \sD_{x_0} \phi (\tau, t_0, x_0, u),
			\label{eq1:tm}\\
			\sD_{x_0} \phi (t_0, t_0, x_0, u) &= \I_n,
			\label{eq2:tm}\\
			\sD_{x_0} \phi (t, t_0, x_0, u)
			&= \sD_x \phi (t, \tau, \phi(\tau, t_0, x_0, u), u)
            \sD_{x_0} \phi (\tau, t_0, x_0, u)
			\label{eq3:tm}
		\end{align}
   	 \end{subequations}
     for each~$(t, \tau) \in \bR \times \bR$ such that~$t_0 \le \tau \le t \le t_1$. Thus, the solution to the variational system at each time~$t \in [t_0, t_1]$ is
\begin{align}\label{eq:vsys_sol}
	\delta x(t) 
    = \sD_{x_0} \phi (t, t_0, x_0, u) \delta x_0
	    + \int_{t_0}^{t} \bigl( \sD_{{x_\tau}} \phi (t, \tau, x_\tau, u) \sD_u f(\tau, x_\tau, u(\tau))\bigr)\bigr|_{x_\tau = \phi (\tau, t_0, x_0, u)} \delta u (\tau)\, d\tau
\end{align}
for every initial condition~$\delta x_0 \in \bR^n$ and every continuous input~$\delta u:\bR \to\bR^m$.

In the next subsection, we offer insight into the exponential stability of the variational system. As we show, this property guarantees the global existence of solutions.

\begin{lem}[Global Existence of Solution]\label{lem:fc}
    Suppose that Assumption~\ref{asm:vf} holds. Given exponential growth rate~$b\in \bR$ and overshoot constant~$k \in \bR_{\ge 1}$, suppose that for each~$(t_0, x_0) \in \bR \times \bR^n$ and every continuous~$u:\bR \to \bR^m$, 
    \begin{align}\label{eq:exbd}
        \| \sD_{x_0} \phi (t, t_0, x_0, u) \|_\cX \le k \e^{b(t - t_0)}
    \end{align}
    for all~$t \in [t_0, T^+(t_0,x_0,u))$, where~$T^+(t_0,x_0,u)$ is the forward escape time, i.e., $T^+(t_0,x_0,u) := \sup \{t \in \bR_{> t_0} : \phi (\tau, t_0, x_0, u) \mbox{ exists for all } \tau \in [t_0, t]\}$. Then, $T^+(t_0,x_0,u) = \infty$.
\end{lem}
\begin{pf}
    The proof is given in Section~\ref{pf:fc}. \qed
\end{pf}

\subsection{Main Results}
\begin{thm}\label{thm:ISS}
\textbf{\textup{(Input-to-State Stability Type Implications)}}
	Suppose that Assumption~\ref{asm:vf} holds.
    Given state norm~$\| \cdot \|_\cX$, input norm~$\| \cdot \|_\cU$, exponential growth rate~$b\in \bR$, overshoot constant~$k \in \bR_{\ge 1}$, and input Lipschitz constant~$\ell \in \bR_{\ge 0}$, consider the following four properties:

    \begin{enumerate}[label=(\arabic*)]
        \item\label{i1:ISS} (Infinitesimal Contraction Condition) 
        there exists a scalar-valued function~$V(t,x,\delta x)$, jointly lower semicontinuous in~$(t,x,\delta x) \in \bR \times \bR^n \times \bR^n$, such that
        \begin{subequations}\label{eq:ICC}
            \begin{align}
                &\| \delta x \|_\cX \le V(t,x,\delta x) \le k \| \delta x \|_\cX, \label{eq:ICC_bd}\\
                &\limsup_{s \to 0^+} \frac{V(t+s, \phi (t+s, t, x, u), \sD_x \phi (t+s, t, x, u) \delta x) - V(t,x,\delta x)}{s} \le b V(t,x,\delta x)  \label{eq:ICC_dini}
            \end{align}
        \end{subequations}
        for all~$(t, x, \delta x) \in \bR \times \bR^n \times \bR^n$ and continuous~$u:\bR \to \bR^m$;


        \item\label{i2:ISS} (Incremental Lyapunov Condition) 
        there exists a scalar-valued function~$W(t,x,y)$, jointly lower semicontinuous in~$(t,x,y) \in \bR \times \bR^n \times \bR^n$, such that
        \begin{subequations}\label{eq:ILC}
            \begin{align}
                &\| x - y \|_\cX \le W(t,x,y) \le k \| x - y \|_\cX, \label{eq:ILC_bd}\\
                &\limsup_{s \to 0^+} \frac{W(t+s, \phi (t+s, t, x, u),  \phi (t+s, t, y, u))-W(t,x,y) }{s} \le b W(t,x,y)  \label{eq:ILC_dini}
            \end{align}
        \end{subequations}
        for all~$(t, x, y) \in \bR \times \bR^n \times \bR^n$ and continuous~$u:\bR \to \bR^m$;

    
        \item\label{i3:ISS} (ISS Type Property of Variational System~\eqref{eq:vsys})
       	    \begin{align}\label{eq:UISS}
                \| \delta x(t) \|_{\cX}
            	\le k \e^{b (t - t_0)} \| \delta x_0 \|_\cX
                + k \ell \int_{t_0}^t \e^{b (t - \tau)} \|\delta u(\tau) \|_\cU \, d\tau
            \end{align}
            for all~$(t, t_0) \in \bR_{\ge t_0} \times \bR$, $(x_0, \delta x_0) \in \bR^n \times \bR^n$, and continuous $(u, \delta u): \bR \to \bR^m \times \bR^m$;

        
        \item\label{i4:ISS} (Incremental ISS Type Property of System~\eqref{eq:sys}) 
            \begin{align}\label{eq:IISS}
                \| \phi  (t, t_0,  y_0, v) - \phi  (t, t_0, x_0, u) \|_\cX
           		\le k \e^{b (t - t_0)} \| y_0 - x_0 \|_\cX
            	    + k \ell \int_{t_0}^t \e^{b (t - \tau)} \| v(\tau) - u(\tau) \|_\cU \, d\tau
       		\end{align}
        	for all~$(t, t_0) \in \bR_{\ge t_0} \times \bR$, $(x_0, y_0) \in \bR^n \times \bR^n$, and continuous $(u, v): \bR \to \bR^m \times \bR^m$.
    	\end{enumerate}
        Then, we have the implications in Fig.~\ref{fig:main}.
\end{thm}

\begin{pf}
    The proof is given in Section~\ref{pf:ISS}. \qed
\end{pf}

\begin{figure*}[h]
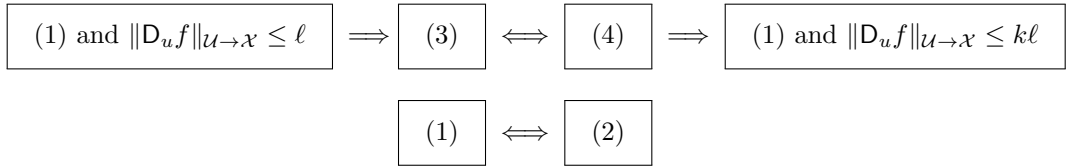

\centering
\begin{align*}
\boxed{\;\ref{i1:ISS} \text{ and } \|\sD_u f\|_{\ux} \le \ell\;}
\;\Longrightarrow\;
&\boxed{\;\ref{i3:ISS}\;}
\;\Longleftrightarrow\;
\boxed{\;\ref{i4:ISS}\;}
\;\Longrightarrow\;
\boxed{\;\ref{i1:ISS} \text{ and } \|\sD_u f\|_{\ux} \le k\ell\;}
\\[3mm]
&\boxed{\;\ref{i1:ISS}\;}
\;\Longleftrightarrow\;
\boxed{\;\ref{i2:ISS}\;}
\end{align*}
\caption{Logical relations established in Theorem~\ref{thm:ISS} among the
infinitesimal contraction condition~\ref{i1:ISS}, the incremental Lyapunov condition~\ref{i2:ISS}, the ISS-type property of the variational system~\ref{i3:ISS}, and the incremental ISS-type property of the
original system~\ref{i4:ISS}. The input-gain bound is recovered up to the
overshoot factor~$k$ (from~$\ell$ on the left to~$k\ell$ on the right).}
\label{fig:main}
\end{figure*}

\begin{figure*}[h]
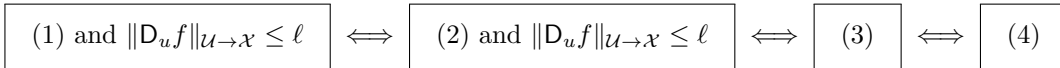

\centering
\begin{align*}
\boxed{\;\ref{i1:ISS} \text{ and } \|\sD_u f\|_{\ux} \le \ell\;}
\;\Longleftrightarrow\;
\boxed{\;\ref{i2:ISS} \text{ and } \|\sD_u f\|_{\ux} \le \ell\;}
\;\Longleftrightarrow\;
\boxed{\;\ref{i3:ISS}\;}
\;\Longleftrightarrow\;
\boxed{\;\ref{i4:ISS}\;}
\end{align*}
\caption{Logical equivalences established in Theorem~\ref{thm:ISS} when~$k=1$ among the
infinitesimal contraction condition~\ref{i1:ISS}, the incremental Lyapunov condition~\ref{i2:ISS}, the ISS-type property of the variational system~\ref{i3:ISS}, and the incremental ISS-type property of the
original system~\ref{i4:ISS}.}
\label{fig2:main}
\end{figure*}

When Theorem~\ref{thm:ISS} is applied to~$k=1$, we have the equivalences in Fig.~\ref{fig2:main}. 
Moreover, $V(t, x, \delta x)=\| \delta x \|_\cX$, and \eqref{eq:ICC_dini} reduces to the norm-based contractivity condition~\cite[Corollary 3.17]{FB:26}:
    \begin{align}\label{eq:lgn}
        &\mu_\cX (\sD_x f) := \sup_{(t,x,u) \in \bR \times \bR^n \times \bR^m} \mu_\cX (\sD_x f(t,x,u)) \le b \\
        &\quad \mu_\cX (\sD_x f(t,x,u)) := \lim_{h \to 0^+} \frac{\|\I_n + h\sD_x f (t, x, u)\|_\cX - 1}{h} 
        = \lim_{h \to 0^+} \frac{\|\sD_x \phi (t+h, t, x, u)\|_\cX - 1}{h}. \nonumber
    \end{align}

\begin{rem}[ISS Type Property]\label{rem:ISS}
    We refer to~\eqref{eq:UISS} and~\eqref{eq:IISS} as ISS-type properties because they relate to ISS. When~$b = -c \le 0$, the integral~$\int_{t_0}^t \e^{-c (t - \tau)} \| v(\tau) - u(\tau) \|_\cU \, d\tau$ admits two upper bounds:
    \begin{itemize}
        \item the signal~$L_\infty$-norm: $\frac{1}{c}\sup_{\tau \in [t_0,t]} \| v(\tau) - u(\tau) \|_\cU$;
        \item the signal~$L_1$-norm: $\int_{t_0}^t \| v(\tau) - u(\tau) \|_\cU \, d\tau$. 
    \end{itemize}
     Consequently,~\eqref{eq:IISS} implies both incremental ISS and integral ISS. For~$b > 0$, in contrast, the upper bound growth unbounded as~$t (> t_0)$ increases and gives a bound on the divergence rate over each finite time interval. Similar statements hold for~$\int_{t_0}^t \e^{-c (t - \tau)} \| \delta u(\tau) \|_\cU \, d\tau$.
    \red
\end{rem}

\begin{rem}[Boundedness of Solution]\label{rem:bd}
    By Lemma~\ref{lem:fc}, Theorem~\ref{thm:ISS} shows that if either of items~\ref{i1:ISS} -- \ref{i4:ISS} holds, then the solution~$\phi(t,t_0, x_0, u)$ exists and is unique for each~$(t,t_0,x_0) \in \bR_{\ge t_0} \times \bR \times \bR^n$ and every continuous~$u:\bR \to \bR^m$. Moreover, from the proof of Lemma~\ref{lem:fc}, the solution is bounded, i.e.,~$\limsup_{t \to \infty}\|\phi(t,t_0, x_0, u)\|_\cX < \infty$ if
    \begin{align}
        \limsup_{t \to \infty} \int_{t_0}^t \e^{b (t -\tau)} \| f(\tau,x_0,u(\tau)) \|_\cX \, d\tau < \infty.
    \end{align}
    This extends the local infinitesimal contractivity analysis~\cite[E3.13]{FB:26} from autonomous systems to time-varying and control systems.
    \red
\end{rem}

\begin{rem}[Continuity of~$V$ and~$W$]
    From the constructions of~$V$ and~$W$ in the proof of Theorem~\ref{thm:ISS}, items~\ref{i3:ISS} and~\ref{i4:ISS} do not necessary imply their continuity. However, the continuity can be guaranteed if i) we restrict the range of the inputs~$u$ and~$v$ into a compact set~$U \subset \bR^m$, and ii) we allow to obtain slightly conservative bounds in~\eqref{eq:ICC_dini} and~\eqref{eq:ILC_dini} by replacing~$b$ with an arbitrary constant~$a \in \bR_{>b}$. This is because, in this case, the supremums in~\eqref{pf4:ISS} and~\eqref{pf9:ISS} can be taken in compact sets of~$h$ and~$u(t) \in U$. \red
\end{rem}

\subsection{Exponential Growth Estimations}
As a special case, consider a non-control system:
\begin{align}\label{eq:sys_no_u}
	\dot x = g(t,x).
\end{align}
Its solution at time~$t \in \bR$ with the initial condition~$x(t_0) = x_0 \in \bR^n$ is denoted by~$x(t) = \psi (t, t_0, x_0)$. The variational system of~\eqref{eq:sys_no_u} along~$\psi (t, t_0, x_0)$ with the initial condition~$\delta x(t_0) = \delta x_0 \in \bR^n$ is
\begin{align}\label{eq:vsys_no_u}
	\dot {\delta x} = \bigl( \sD_x g(t, x_t) \delta x \bigr) \bigr|_{x_t =\psi (t, t_0, x_0)}.
\end{align} 
Theorem~\ref{thm:ISS} with~$\ell = 0$ yields the following corollary. 

\begin{cor}\label{cor:IES}
\textbf{\textup{(Exponential Growth Equivalences)}}
    Let the vector field~$g(t,x)$ and its Jacobian~$\sD_x g(t,x)$ exist and be jointly continuous in~$(t,x) \in \bR \times \bR^n$. Given state norm~$\| \cdot \|_\cX$, exponential growth rate~$b\in \bR$, and overshoot constant~$k \in \bR_{\ge 1}$, the following four properties are equivalent:
	
    \begin{enumerate}[label=(\arabic*)]
        \item\label{i1:IES} (Infinitesimal Contraction Condition)
            there exists a scalar-valued function~$V(t,x,\delta x)$, jointly lower semicontinuous in~$(t,x,\delta x) \in \bR \times \bR^n \times \bR^n$, such that
        \begin{subequations}
            \begin{align*}
                &\| \delta x \|_\cX \le V(t,x,\delta x) \le k \| \delta x \|_\cX, \\
                &\limsup_{s \to 0^+} \frac{V(t+s, \psi (t+s, t, x), \sD_x \psi (t+s, t, x) \delta x) - V(t,x,\delta x)}{s} \le b V(t,x,\delta x)
            \end{align*}
        \end{subequations}
        for all~$(t, x, \delta x) \in \bR \times \bR^n \times \bR^n$;


        \item\label{i2:IES} (Incremental Lyapunov Condition) 
        there exists a scalar-valued function~$W(t,x,y)$, jointly lower semicontinuous in~$(t,x,y) \in \bR \times \bR^n \times \bR^n$, such that
        \begin{subequations}
            \begin{align*}
                &\| x - y \|_\cX \le W(t,x,y) \le k \| x - y \|_\cX, \\
                &\limsup_{s \to 0^+} \frac{W(t+s, \psi (t+s, t, x),  \psi (t+s, t, y))-W(t,x,y) }{s} \le b W(t,x,y)
            \end{align*}
        \end{subequations}
        for all~$(t, x, y) \in \bR \times \bR^n \times \bR^n$;
        
        
        \item\label{i3:IES} (Exponential Estimation of Variational System~\eqref{eq:vsys})
            \begin{align*}
                \| \delta x(t) \|_\cX \le k \e^{b (t - t_0)} \| \delta x_0 \|_\cX
            \end{align*}
            for all~$(t, t_0) \in \bR_{\ge t_0} \times \bR$ and~$(x_0, \delta x_0) \in \bR^n \times \bR^n$;
        			

        \item\label{i4:IES}  (Incremental Exponential Estimation of System~\eqref{eq:sys})
             \begin{align}\label{eq:IES}
                \| \psi (t, t_0, y_0) - \psi (t, t_0, x_0) \|_\cX \le k \e^{b (t - t_0)} \| y_0 - x_0 \|_\cX
            \end{align}
            for all~$(t, t_0) \in \bR_{\ge t_0} \times \bR$ and~$(x_0, y_0) \in \bR^n \times \bR^n$. \red
        \end{enumerate}
\end{cor}

From Remark~\ref{rem:bd}, if i)~$b < 0$, and ii)~$g(t,x)$ is independent of~$t$, then the solution~$\psi (t, t_0, x_0)$ is bounded, which implies the existence of an equilibrium. We include this argument for the sake of self-containedness. Since~$\psi (t, t_0, x_0)$ is bounded, its $\omega$-limit set~$\omega (x_0)$ is compact. Fix~$z \in \omega(x_0)$ and take a sequence~$\{t_n\}$ such that~$\psi (t_n, t_0, x_0) \to z$ as~$n \to \infty$. Take~$y_0 = \psi (s, t_0, x_0)$ for an arbitrary~$s \in \bR_{\ge t_0}$. Then, the semigroup property and continuity of the flow and the time-invariance of the system yield
\begin{align*}
   \| \psi (s, t_0, z) - z \|_\cX 
   &= \lim_{n \to \infty} \| \psi (s, t_0, \psi (t_n, t_0, x_0)) - \psi (t_n, t_0, x_0) \|_\cX \\
   &=\lim_{n \to \infty} \| \psi (t_n, t_0, \psi (s, t_0, x_0)) - \psi (t_n, t_0, x_0) \|_\cX \overset{\eqref{eq:IES}}{=} 0.
\end{align*} 
Thus, $z$ is an equilibrium as this holds for each~$s \in \bR_{\ge t_0}$. Substituting~$y_0 = z$ into~\eqref{eq:IES} shows the global exponential stability of~$z$, and thus the equilibrium is unique. This generalizes~\cite[Theorem 3.9]{FB:26} for the overshoot constant~$k=1$ to arbitrary~$k \in \bR_{\ge 1}$.

\subsection{Comparative Analyses}\label{sec:ca}
We provide comparative analyses of this work with existing literature.

{\it Exponential Growth Estimation (Corollary~\ref{cor:IES}):}\vspace{-3mm}
\begin{itemize}
    \item The equivalences of items~\ref{i1:IES}, \ref{i2:IES}, and~\ref{i4:IES} can be found in~\cite[Theorem 29]{AD-SJ-FB:22} when the overshoot constant is~$k=1$ (contractivity with respect to a norm).

    \item The equivalence of items~\ref{i3:IES} and~\ref{i4:IES} can be found in~\cite[Proposition 1]{NB-RO-AP:19} by assuming the global existence of the solution. However, this does not show the implication for the same~$b$ and~$k$.

    \item The equivalences of items~\ref{i1:IES}, \ref{i3:IES} and~\ref{i4:IES} are established when norm~$\| \cdot \|_\cX$ is the~$\ell_2$-norm~\cite[Propositions~4]{VA-BJ-LP:16} and, for cooperative systems, are the~$\ell_1$- and~$\ell_\infty$-norms~\cite[Theorem 4.3]{YK-BB-MC:20} in the time-invariant case. Moreover, they show that~$V(x,\delta x)$ can be selected as quadratic~$(\delta x^\top P(x) \delta x)^{1/2}$ with symmetric and positive definite matrix-valued function~$P(x)$, sum-separable~$\sum_i v_i (x) \delta x_i$, and max-separable~$\max_i \{ \delta x_i/w_i(x) \}$ with element-wise positive vector-valued functions~$v(x)$ and~$w(x)$, respectively, by assuming i)~global existence of the solution, ii)~$b < 0$, iii)~$g$ is time independent and twice continuously differentiable, and iv)~$\sD_x g$ is bounded. However, they do not show the implications for the same~$b$ and~$k$.

    \item Similar equivalences as of items~\ref{i2:IES} and~\ref{i4:IES} can be found for incremental asymptotic stability in~\cite[Theorem~1]{DA:02} and~\cite[Theorem 2]{DA:09}.

    \item The equivalence of items~\ref{i1:IES} and~\ref{i2:IES} can be found in~\cite[Theorem 3.1]{YK-BB:24} by assuming i)~global existence of the solution, and ii)~$g$ is time independent and twice continuously differentiable.
\end{itemize}
Corollary~\ref{cor:IES} establishes the equivalences with respect to an arbitrary norm~$\| \cdot \|_\cX$, for the same~$b$ and~$k$, under Assumption~\ref{asm:vf} only. Note that the implication: item~\ref{i1:IES} $\implies$ item~\ref{i3:IES} implies that~$V(t,x,\delta x)$ is a Lyapunov function for the variational system~\eqref{eq:vsys} (at~$\delta u=\0_m$).

{\it Input-to-State Stability Type Implications (Theorem~\ref{thm:ISS}):}\vspace{-3mm}
\begin{itemize}
    \item The implication: (item~\ref{i1:ISS} and $\|\sD_u f\|_{\ux} \le \ell$) $\implies$ item~\ref{i4:ISS} is shown in~\cite[Corollary 3.17]{FB:26} for~$k=1$. For incremental ISS of~$\dot x = \bar f(x) + \bar g(x) u$ when~$\| \cdot \|_\cX$ is the~$\ell_2$-norm, this implication is established in~\cite[Theorem 2]{MG-DA-VA:23} by assuming that i)~$b<0$, ii)~$\bar f$ and~$\bar g$ are twice continuously differentiable, and iii)~$\bar g$ is bounded.
   
    \item The implication: item~\ref{i4:ISS} $\implies$ item~\ref{i3:ISS} is established for incremental ISS of~$\dot x = \bar f(x) + \bar g(x) u$, under the assumption that the first- and second-order derivatives of~$\bar f$ and~$\bar g$ exist and are bounded~\cite[Proposition~1]{MG-DA-VA:23}.

    \item Similar equivalences as of items~\ref{i2:ISS} and~\ref{i4:ISS} can be found in~\cite[Theorem 2]{DA:02} for incremental ISS and in~\cite[Theorem 1]{DA:09} for incremental integral ISS.
\end{itemize}
Theorem~\ref{thm:ISS} establishes the equivalences with respect to an arbitrary state norm~$\| \cdot \|_\cX$ and input norm~$\| \cdot \|_\cU$ under Assumption~\ref{asm:vf} only. More importantly, to the best of the authors' knowledge, the converse implication: item~\ref{i3:ISS} or~\ref{i4:ISS}~$\implies$ (item~\ref{i1:ISS} and $\|\sD_u f\|_{\ux} \le k \ell$) is not established in the literature. As a byproduct, we show that item~\ref{i3:ISS} or~\ref{i4:ISS} implies boundedness of~$\sD_u f$, and that boundedness of~$\sD_x f$ is not required in our analysis. Finally, by combining Theorem~\ref{thm:GD} with Remark~\ref{rem:ISS}, the chain of implications (item~\ref{i1:ISS} and $\|\sD_u f\|_{\ux} \le \ell$) $\iff$ (item~\ref{i2:ISS} and $\|\sD_u f\|_{\ux} \le \ell$) $\implies$ (item~\ref{i3:ISS}) $\iff$ (item~\ref{i4:ISS}) extends to incremental ISS and incremental integral ISS.



\section{Applications}

Corollary~\ref{cor:IES} with the overshoot constant~$k=1$ sheds light onto related existing concepts in
dynamical systems because they now can be understood in terms of variational systems.

\subsection{Sensitivity Matrix}
Given an initial condition $(t_0,x_0) \in \bR \times \bR^n$, the sensitivity matrix
of a nonlinear system~\eqref{eq:sys_no_u} at time $t \in \bR_{\ge t_0}$ is given by
    \begin{align*}
	   S(t, t_0, x_0) := \sD_{x_0} \psi (t, t_0, x_0).
    \end{align*}
    From~\eqref{eq:tm}, this is the transition matrix of the variational system~\eqref{eq:vsys_no_u}, i.e.,~$\delta x(t) = S(t, t_0, x_0) \delta x_0$.
    
    Suppose that~$b := \sup_{(t,x) \in \bR \times \bR^n}\mu_\cX (\sD_xg (t,x))$ is finite, where recall~\eqref{eq:lgn} for the definition of the log norm~$\mu_\cX$.
    According to the implication: item~\ref{i1:IES} $\implies$ item~\ref{i3:IES} in Corollary~\ref{cor:IES} for the overshoot constant~$k=1$, we have
    \begin{align}\label{eq:sns}
    	\| S(t, t_0, x_0) \|_\cX  = \| \sD_{x_0} \psi (t, t_0, x_0) \|_\cX   \le \e^{b (t - t_0)}, \quad \forall (t, t_0, x_0) \in \bR_{\ge t_0} \times \bR \times \bR^n.
    \end{align}
    This analysis can be generalized to a control system~\eqref{eq:sys}
    because Theorem~\ref{thm:GD} shows that the
    solution to the variational system~\eqref{eq:vsys} can be viewed as a
    sensitivity matrix with respect to small perturbations~$(x_0, u)
    \mapsto (x_0 + h \delta x_0, u + h \delta u)$.

    The sensitivity bound~\eqref{eq:sns} is closely related
      to the analysis on non-exploding gradients in contractive neural
      differential equations obtained in~\cite[Theorem~2]{MZ-LX-GFT:23}.

    \subsection{Lyapunov Characteristic Exponent}
    For a nonlinear system~\eqref{eq:sys_no_u}, the \emph{Lyapunov characteristic exponent}~\cite[Chapter 29.1]{SW:03} is defined by using the solution to its variational system~\eqref{eq:vsys_no_u} as follows:
    \begin{align*}
        \chi (t_0, x_0, \delta x_0)
        := \limsup_{t \to \infty} \frac{1}{t} 
            \ln \left( \frac{\| \sD_{x_0} \psi (t, t_0, x_0) \delta x_0\|_\cX}{\| \delta x_0 \|_\cX} \right).
    \end{align*}
    If~$b:= \sup_{(t,x) \in \bR \times \bR^n} \mu_\cX (\sD_x g(t,x))$ is finite, we have
    \begin{align*}
        \chi (t_0, x_0, \delta x_0) \le b,
        \quad
        \forall (t_0, x_0, \delta x_0) \in \bR \times \bR^n \times \bR^n.
    \end{align*}
    This can be confirmed as follows:
    \begin{align*}
        \chi (t_0, x_0, \delta x_0)
        \le \limsup_{t \to \infty} \frac{1}{t} \ln \| \sD_{x_0} \psi (t, t_0, x_0) \|_\cX
        \overset{\eqref{eq:sns}}{\le} b.
    \end{align*}
    Next, define the \emph{maximum Lyapunov exponent} by
        \begin{align*}
        \chi_{\max} (t_0, x_0)
        := \limsup_{t \to \infty} \frac{1}{t} \ln \sup_{\delta x_0\neq0}
            \frac{\| \sD_{x_0} \psi (t, t_0, x_0) \delta x_0\|_\cX}{\| \delta x_0 \|_\cX}.
    \end{align*}
    The maximal Lyapunov exponent quantifies the maximal asymptotic exponential
    rate of separation of trajectories starting arbitrarily close to $x_0$.
    The same calculations lead to the bound~$\sup_{(t_0,x_0)} \chi_{\max} (t_0, x_0) \le b$.


\section{Discrete-Time Case}

We investigate the discrete-time version of Theorem~\ref{thm:ISS}. Consider a discrete-time nonlinear control system:
\begin{align}\label{deq:sys}
	x(t+1) = f(t, x(t), u(t)),
\end{align}
where~$f: \bZ \times \bR^n \times \bR^m \to \bR^n$. Let~$\phi(t,t_0, x_0, u)$ denote the solution to the system~\eqref{deq:sys} at time instant~$t \in \bZ_{\ge t_0}$ with initial condition~$x(t_0) = x_0 \in \bR^n$ and input~$u:\bZ \to \bR^m$. Namely, $x(t) = \phi(t,t_0, x_0, u)$. As analogues to the continuous-time case, we impose the following assumption.

\begin{assum}[Continuous Differentiability of Vector Field in Discrete-time]\label{dasm:vf}
  For each~$t \in \bZ$, the vector field~$f(t,x,u)$ and its Jacobians $\sD_x f(t,x,u)$ and~$\sD_u f(t,x,u)$ exist and are jointly continuous in~$(x,u) \in \bR^n \times \bR^m$.~\red
\end{assum}

The variational system of~\eqref{deq:sys} along~$\phi (t, t_0, x_0, u)$ with the initial condition~$\delta x(t_0) = \delta x_0 \in \bR^n$ and input~$\delta u:\bZ \to \bR^m$ is
\begin{align}\label{deq:vsys}
	{\delta x} (t+1) = \bigl(\sD_x f(t, x_t, u(t)) \delta x(t) + \sD_u f(t, x_t, u(t)) \delta u (t)\bigr)\bigr|_{x_t=\phi (t, t_0, x_0, u)}.
\end{align}

In the discrete-time case, an analogue of Theorem~\ref{thm:GD} follows from the chain rule.
\begin{thm}[Solution of Variational system in Discrete-Time]\label{dthm:GD}
	Under Assumption~\ref{dasm:vf}, the solution~$\delta x(t)$ to the variational system~\eqref{deq:vsys} satisfies
    \begin{align}\label{deq:sol_GD}
        \delta x(t) = \left.\frac{\partial \phi (t, t_0, x_0 + h \delta x_0, u + h \delta u)}{\partial h}\right|_{h = 0}
	\end{align}
    for each~$(t, t_0) \in \bZ_{\ge t_0} \times \bZ$, every initial state~$(x_0, \delta x_0) \in \bZ \times \bR^n \times \bR^n$, and every input~$(u, \delta u):\bZ \to \bR^m \times \bR^m$.~\red
\end{thm}

As a counterpart of Theorem~\ref{thm:ISS}, we have the following implications.

\begin{thm}\label{dthm:ISS}
\textbf{\textup{(Input-to-State Stability Type Implications in Discrete-Time)}}
	Suppose that Assumption~\ref{dasm:vf} holds.
    Given state norm~$\| \cdot \|_\cX$, input norm~$\| \cdot \|_\cU$, growth factor~$b\in \bR_{>0}$, overshoot constant~$k \in \bR_{\ge 1}$, and input Lipschitz constant~$\ell \in \bR_{\ge 0}$, consider the following four properties:

    \begin{enumerate}[label=(\arabic*)]
        \item\label{di1:ISS} (Infinitesimal Contraction Condition) 
        there exists a scalar-valued function~$V(t, x, \delta x)$, jointly lower semicontinuous in~$(x,\delta x) \in \bR^n \times \bR^n$ for each~$t \in \bZ$, such that
        \begin{subequations}\label{deq:ICC}
            \begin{align}
                &\| \delta x \|_\cX \le V(t, x, \delta x) \le k \| \delta x \|_\cX, \label{deq:ICC_bd}\\
                &V (t+1, f(t, x, u), \sD_x f(t, x, u) \delta x ) \le b V(t,x,\delta x)  \label{deq:ICC_dini}
            \end{align}
        \end{subequations}
        for all~$(t, x, \delta x, u) \in \bZ \times \bR^n \times \bR^n \times \bR^m$;


        \item\label{di2:ISS} (Incremental Lyapunov Condition) 
        there exists a scalar-valued function~$W(t,x,y)$, jointly lower semicontinuous in~$(x,y) \in \bR^n \times \bR^n$ for each~$t \in \bZ$, such that
        \begin{subequations}\label{deq:ILC}
            \begin{align}
                &\| x - y \|_\cX \le W(t,x,y) \le k \| x - y \|_\cX, \label{deq:ILC_bd}\\
                &W(t+1, f(t, x, u), f(t, y, u)) \le b W(t,x,y)  \label{deq:ILC_dini}
            \end{align}
        \end{subequations}
        for all~$(t, x, y, u) \in \bZ \times \bR^n \times \bR^n \times \bR^m$;

    
        \item\label{di3:ISS} (ISS Type Property of Variational System~\eqref{deq:vsys})
       	    \begin{align}\label{deq:UISS}
                \| \delta x(t) \|_{\cX}
            	\le k b^{t - t_0} \| \delta x_0 \|_\cX
                + k \ell \sum_{\tau = t_0}^{t-1} b^{t - (\tau+1)} \|\delta u(\tau) \|_\cU
            \end{align}
            for all~$(t, t_0) \in \bZ_{\ge t_0} \times \bZ$, $(x_0, \delta x_0) \in \bR^n \times \bR^n$, and $(u, \delta u): \bZ \to \bR^m \times \bR^m$;

        
        \item\label{di4:ISS} (Incremental ISS Type Property of System~\eqref{deq:sys}) 
            \begin{align}\label{deq:IISS}
                \| \phi  (t, t_0,  y_0, v) - \phi  (t, t_0, x_0, u) \|_\cX
           		\le k b^{t - t_0} \| y_0 - x_0 \|_\cX
            	    + k \ell \sum_{\tau = t_0}^{t-1} b^{t - (\tau+1)} \| v(\tau) - u(\tau) \|_\cU
       		\end{align}
        	for all~$(t, t_0) \in \bZ_{\ge t_0} \times \bZ$, $(x_0, y_0) \in \bR^n \times \bR^n$, and $(u, v): \bZ \to \bR^m \times \bR^m$.
    	\end{enumerate}
        Then, we have the implications in Fig.~\ref{fig:main}.
\end{thm}

Similarly to the continuous-time case, when Theorem~\ref{dthm:ISS} is applied to~$k=1$, we have the equivalences in Fig.~\ref{fig2:main}. Moreover, $V(t, x, \delta x)=\| \delta x \|_\cX$, and \eqref{deq:ICC_dini} reduces to~$\|\sD_x f\|_\cX \le b$.

\begin{rem}[Boundedness of Solution in Discrete-Time]\label{drem:bd}
    Suppose that either of items~\ref{di1:ISS}--\ref{di4:ISS} in Theorem~\ref{dthm:ISS} holds. Then, the solution~$\phi(t,t_0,x_0, u)$ is bounded if
    \begin{align}\label{deq:sol_bd}
        \limsup_{t \to \infty} \sum_{\tau = t_0}^{t-1} b^{t - (\tau+1)} \| f(\tau, x_0, u(\tau)) - x_0 \|_\cX < \infty.
    \end{align}
    We show this by using~\eqref{deq:IISS}. It follows that
    \begin{align*}
        \phi (t, t_0, x_0, u) - x_0 
        &= \sum_{\tau = t_0}^{t-1} ( \phi (t, \tau, x_0, u) - \phi (t, \tau + 1, x_0, u) ) \\
        &= \sum_{\tau = t_0}^{t-1} (  \phi (t, \tau+1, f(\tau, x_0, u(\tau)), u) - \phi (t, \tau + 1, x_0, u) ).
    \end{align*}
    Taking the norm and using the triangular inequality give
    \begin{align}\label{deq1:sol_bd}
        \| \phi (t, t_0, x_0, u) - x_0 \|_\cX
        \le \sum_{\tau = t_0}^{t-1} \| \phi (t, \tau+1, f(\tau, x_0, u(\tau)), u) - \phi (t, \tau + 1, x_0, u) \|_\cX.
    \end{align}    
    Substituting~$t_0 = \tau + 1 (\le t)$, $y_0 = f(\tau, x_0, u(\tau))$, and~$u \equiv v$ into~\eqref{deq:IISS} yields
    \begin{align}\label{deq2:sol_bd}
        \| \phi  (t, \tau+1, f(\tau, x_0, u(\tau)), v) - \phi  (t, \tau+1, x_0, u) \|_\cX
        \le k b^{t - (\tau+1)} \| f(\tau, x_0, u(\tau)) - x_0 \|_\cX.    
    \end{align}
    Combining~\eqref{deq1:sol_bd} and~\eqref{deq2:sol_bd} leads to
     \begin{align*}
        \| \phi (t, t_0, x_0, u) - x_0 \|_\cX
        \le k \sum_{\tau = t_0}^{t-1} b^{t - (\tau+1)} \| f(\tau, x_0, u(\tau)) - x_0 \|_\cX.
    \end{align*}
    Thus, the solution is bounded if~\eqref{deq:sol_bd} holds.
    \red
\end{rem}

Selecting~$\ell = 0$ in Theorem~\ref{dthm:ISS} yields the similar equivalences as Corollary~\ref{cor:IES} for the exponential growth estimation. Moreover, if i)~$f(t,x,u)$ depends on~$x$ only, and ii)~$b \in [0, 1)$, then~\eqref{deq:sol_bd} holds, and thus the solution is bounded. Also, the system has a globally exponentially stable equilibrium. This can be viewed as an extension of the Banach contraction theorem~\cite[Theorem 1.6]{FB:26} from the overshoot constant~$k=1$ to an arbitrary~$k \in \bR_{\ge 0}$.

At the end of this section, we provide comparative analyses of Theorem~\ref{dthm:ISS} with existing literature.
\begin{itemize}
    \item When~$\ell = 0$, the equivalences of items~\ref{di1:ISS}, \ref{di2:ISS} and~\ref{di4:ISS} are found when norm~$\| \cdot \|_\cX$ is the~$\ell_2$-norm~\cite[Theorems 11 and 15]{DNT-BSR-CMK:19} and~\cite[Corollary 3.8]{YK-YH:24}, and, for cooperative systems, are the~$\ell_1$- and~$\ell_\infty$-norms~\cite[Corollary 6.4 and Theorem 6.5]{YK-YH:25}, where the invertibility of~$f$ with respect to~$x$ is further required for the~$\ell_\infty$-norm. Although they show that if~$b \in (0, 1)$,~$V(t, x,\delta x)$ can be selected as quadratic, sum-separable, and max-separable, respectively, they do not show the implications for the same~$b$ and~$k$. The equivalence of items~\ref{di1:ISS} and~\ref{di2:ISS} for the same~$b$ and~$k$ can be found in~\cite[Corollary 3.6]{YK-BB:24}.

    \item A similar equivalences as of items~\ref{di2:ISS} and \ref{di4:ISS} can be found for incremental ISS (with the signal~$L_\infty$-norm) in~\cite[Theorem 24]{DNT-BSR-CMK:16}.
\end{itemize}
Theorem~\ref{dthm:ISS} establishes the implications for the same~$b$ and~$k$ with respect to an arbitrary state norm~$\| \cdot \|_\cX$ and input norm~$\| \cdot \|_\cU$. To the best of the authors’ knowledge, for the considered ISS-type properties, this is the first paper that establishes the converse implications: item~\ref{di3:ISS} or~\ref{di4:ISS} $\implies$ item~\ref{di1:ISS} or~\ref{di2:ISS} and $\|\sD_u f\|_{\ux} \le k\ell$.


\section{Proofs}\label{sec:pf}
\subsection{Proof of Theorem~\ref{thm:GD}}\label{pf:GD}
    For the sake of simplicity, write~$\phi (t) := \phi (t, t_0, x_0, u)$
    and~$\phi_h (t) := \phi (t, t_0, x_0 + h \delta x_0, u + h \delta u)$,
    and define
    \begin{align}\label{pf1:GD}
        z_h(t) := \frac{\phi_h(t) - \phi(t)}{h}, \quad h > 0.
    \end{align}
    From the integral representations of~$\phi(t)$ and~$\phi_h(t)=\phi(t)+h z_h(t)$, we have
    \begin{align}\label{pf2:GD}
        z_h(t)
        &= \delta x_0
        + \int_{t_0}^{t} \frac{1}{h} \bigl(
            f(\tau, \phi_h(\tau), u(\tau) + h \delta u (\tau)) - f(\tau, \phi(\tau), u(\tau))
        \bigr) d\tau
        \nonumber\\
        &=
        \delta x_0
        + \int_{t_0}^{t} \frac{1}{h} \bigl(
            f(\tau, \phi(\tau) + h z_h(\tau), u(\tau) + h \delta u (\tau))
       - f(\tau, \phi(\tau), u(\tau))
        \bigr) \, d\tau.
    \end{align}
    Note that for each~$\delta x_0$ and every continuous~$\delta u:\bR \to \bR^m$, there exists~$t_1 \in \bR_{> t_0}$ and~$h_1 \in \bR_{>0}$ such that~$\phi_h(t)$ and thus~$z_h(t)$ exist for all~$t \in [t_0, t_1]$ and~$h \in (0, h_1]$.
    We show~$z_h(t_1) \to \delta x(t_1)$ as~$h \to 0^+$ in two steps.

    {\bf Step~1: Uniform bound on~$z_h$.}
    Let~$\| \cdot \|$ denote a vector norm and its induced matrix norm.
    From Assumption~\ref{asm:vf}, $\phi(t)$ is continuous on~$[t_0, t_1]$.
    Fix~$R > 0$, and define a compact tubular neighborhood:
    \begin{align*}
        \cK_R := \bigl\{
            (t, y, v) \in \bR \times \bR^n \times \bR^m :
            t \in [t_0, t_1], 
            \| y - \phi(t) \| \le R,
             \| v - u(t) \| \le R \bigr\}.
    \end{align*}
    By joint continuity of~$\sD_x f(t,x,u)$ and~$\sD_u f(t,x,u)$ on the
    compact set~$\cK_R$, there exists~$M_R \ge 0$ such that
    \begin{align}\label{pf3:GD}
        \max \{ \| \sD_x f(t, x, u) \|, \| \sD_u f(t, x, u) \| \}
        \le M_R
    \end{align}
    for all~$(t,x,u) \in \cK_R$.
    Define the first exit time
    \begin{align*}
        T_h (\delta x_0, \delta u)
        := \inf \bigl\{
            t \in [t_0, t_1] : (t, \phi_h(t), u(t) + h \delta u(t)) \notin \cK_R \bigr\}
    \end{align*}
    with~$T_h (\delta x_0, \delta u) := t_1$ if no exit occurs. For
    each~$t \in [t_0, T_h (\delta x_0, \delta u)]$, applying the triangle
    inequality and the mean value theorem to~\eqref{pf2:GD} yields
    \begin{align*}
        \| z_h(t) \|
        \le
        \| \delta x_0 \|
        + M_R \int_{t_0}^{t}
            \bigl( \|z_h(\tau)\| + \|\delta u (\tau)\| \bigr) \, d\tau.
    \end{align*}
    By~\eqref{pf1:GD} and the Gr\"onwall--Bellman inequality, we obtain
    \begin{align}\label{pf4:GD}
        \| z_h(t) \| \le \beta (t_1,R), \quad
        \forall t \in [t_0, T_h (\delta x_0, \delta u)]
    \end{align}
    for all~$\| \delta x_0 \| \le R$ and $\sup_{\tau \in [t_0, t_1]} \| \delta u(\tau) \| \le R$, where
    \begin{align}\label{pf5:GD}
        \beta (t_1,R) :=  R \e^{M_R (t_1 - t_0)} (  1  + M_R (t_1 - t_0) )
    \end{align}
    is independent of~$h$. For all~$t \in [t_0, T_h(\delta x_0,\delta u)]$ and~$h \in (0, h^*(t_1,R)]$ with~$h^*(t_1,R):= \min\{R/\beta (t_1,R), h_1\}$,
    \begin{align*}
        \|\phi_h(t) - \phi(t)\|
        \overset{\eqref{pf1:GD}}{=} h\|z_h(t)\| 
        \overset{\eqref{pf4:GD}}{\le} h \beta(t_1,R) 
        \le h^*(t_1,R) \beta(t_1,R) 
        = R
    \end{align*}
    and
    \begin{align}\label{pf6:GD}
        \|h \delta u(t)\|
        \le h^*(t_1,R) R
         = \frac{R^2}{\beta(t_1,R)} \overset{\eqref{pf5:GD}}{\le} R.
    \end{align}
    Thus, $h \in (0, h^*(t_1,R)]$ implies~$T_h(\delta x_0, \delta u) = t_1$, and
    \begin{align}\label{pf7:GD}
        \| z_h(t) \| \le \beta (t_1,R), \quad \forall t \in [t_0, t_1].
    \end{align}

    {\bf Step~2: Convergence~$z_h(t_1) \to \delta x(t_1)$.}
    Define~$e_h(t) := z_h(t) - \delta x(t)$, where~$e_h(t_0) = 0$.
    Subtracting the integral form of~\eqref{eq:vsys} from~\eqref{pf2:GD}, we obtain
    \begin{align}\label{pf8:GD}
        e_h(t)
        =  \int_{t_0}^{t} \bigl( \sD_x f(\tau,\phi(\tau),u(\tau)) e_h(\tau) + r_h(\tau) \bigr) \, d\tau,
    \end{align}
    where~$t \in [t_0, t_1]$ and
    \begin{align*}
        r_h(\tau) 
        &:= \frac{1}{h} \bigl( f(\tau, \phi(\tau) + h z_h(\tau), u(\tau) + h \delta u (\tau)) - f(\tau, \phi(\tau), u(\tau)) \bigr) \\
        &\qquad - \sD_x f(\tau,\phi(\tau),u(\tau)) z_h(\tau) 
            - \sD_u f(\tau,\phi(\tau),u(\tau)) \delta u(\tau).
    \end{align*}
    We show~$r_h(\tau) \to 0$ as~$h \to 0^+$ uniformly in~$\tau \in [t_0,t_1]$. 
    
    The fundamental theorem of calculus yields
    \begin{align*}
        &\frac{1}{h} \bigl( f(\tau,\phi(\tau)+hz_h(\tau), u(\tau)+h\delta u(\tau))
              - f(\tau,\phi(\tau),u(\tau)) \bigr)
        \\
        &= \int_0^1 \bigl(
            \sD_x f(\tau, \phi(\tau) + s h z_h(\tau), u(\tau) + s h\delta u(\tau)) z_h(\tau)
            + \sD_u f(\tau, \phi(\tau) + s h z_h(\tau), u(\tau) + s h\delta u(\tau)) \delta u(\tau)
        \bigr) \; ds.
    \end{align*}
    Using this, $r_h(\tau)$ can be rearranged as
    \begin{align}\label{pf9:GD}
        r_h(\tau)
        &=
        \int_0^1
        \Bigl(
            \bigl(
                \sD_x f(\tau,\phi(\tau) + s h z_h(\tau), u(\tau) + s h \delta u(\tau)) 
                - \sD_x f(\tau, \phi(\tau), u(\tau)) \bigr) z_h(\tau) \nonumber\\
                &\qquad \qquad + \bigl(
                \sD_u f(\tau,\phi(\tau) + s h z_h(\tau), u(\tau) + s h \delta u(\tau))  - \sD_u f(\tau, \phi(\tau), u(\tau)) \bigr) \delta u(\tau)
        \Bigr) \; ds.
    \end{align}
    
    Since~$\sD_x f(t,x,u)$ and~$\sD_u f(t,x,u)$ are jointly continuous on compact~$\cK_R$, they are uniformly continuous on~$\cK_R$. From~\eqref{pf7:GD},
    $(\tau, \phi(\tau) + s h z_h(\tau), u(\tau) + s h\delta u(\tau)) \in \cK_R$
    for all~$\tau \in [t_0,t_1]$, $s \in [0,1]$, and~$h \in (0,h^*(t_1,R)]$. Thus, their common modulus of continuity 
    \begin{align*}
        \omega_R(\sigma) 
         := \sup \Bigl\{ \max \bigl\{
            \|\sD_x f(t,x,u) - \sD_x f(t,x',u')\|, \|\sD_u f(t,x,u) - \sD_u f(t,x',u')\| \bigr\} : \qquad \\
         (t,x,u),(t,x',u') \in \cK_R, \|(x-x', u-u')\| \le \sigma \Bigr\}
    \end{align*}
    satisfies~$\omega_R(\sigma) \to 0$ as~$\sigma \to 0^+$.

    Let~$(x,u) = (\phi(\tau) + s h z_h(\tau), u(\tau) + s h \delta u(\tau))$ and~$(x', u') = (\phi(\tau), u(\tau))$. It follows from~\eqref{pf6:GD} and~\eqref{pf7:GD} that, for all~$s \in [0,1]$,
    \begin{align}\label{pf10:GD}
        \|(sh z_h(\tau),  sh \delta u(\tau))\| 
        \le h (\beta(t_1,R) + R ).
    \end{align}
    Thus, for all~$\tau \in [t_0,t_1]$ and~$h \in (0, h^*(t_1,R)]$, we have
    \begin{align}\label{pf11:GD}
        \|r_h(\tau)\| 
        &\overset{\eqref{pf9:GD},\eqref{pf10:GD}}{\le} \omega_R( h (\beta(t_1,R) + R ))  (\| z_h(\tau) \| + \| \delta u(\tau) \|)\nonumber\\
        &\overset{\eqref{pf7:GD}, \|\delta u(\tau)\| \le R}{\le} \omega_R( h (\beta(t_1,R) + R )) (\beta(t_1,R) + R) 
        =: \rho_h,
    \end{align}
    where~$\rho_h \to 0$ as~$h \to 0^+$, uniformly in~$\tau \in [t_0,t_1]$.
    
    Finally, it follows that
    \begin{align*}
        \| e_h(t) \| 
        \overset{\eqref{pf3:GD},\eqref{pf8:GD},\eqref{pf11:GD}}{\le} M_R \int_{t_0}^t \| e_h(\tau) \| \; d\tau + \rho_h (t - t_0).
    \end{align*}
    By the Gr\"onwall--Bellman inequality, we obtain
    \begin{align}\label{pf12:GD}
        \| e_h(t) \| \le \rho_h (t - t_0) \operatorname{e}^{M_R (t-t_0)},
        \quad \forall t \in [t_0, t_1].
    \end{align}
    Since $\rho_h \to 0$ as~$h \to 0^+$, uniformly in~$\tau \in [t_0,t_1]$, evaluating this at~$t = t_1$ gives~$\| e_h(t_1) \| \to 0$, i.e., $z_h(t_1) \to \delta x(t_1)$ as~$h \to 0^+$.
    \qed


\subsection{Proof of Lemma~\ref{lem:fc}}\label{pf:fc}
Fix~$x_0 \in \bR^n$ and continuous~$u:\bR \to \bR^m$, define~$\phi^u (t, \tau, x_0) := \phi (t, \tau, x_0, u)$. From the semigroup property of the flow, $\phi^u(t,s,\phi^u(s,\tau,x_0))=\phi^u(t,\tau,x_0)$ holds as long as it exists. Taking its partial derivative with respect to~$s$ and evaluating it at~$s=\tau$ yield
\begin{align}
    \frac{\partial \phi^u(t,\tau,x_0)}{\partial \tau} 
    = - \sD_{x_0}\phi^u(t,\tau,x_0) f(\tau,x_0,u(\tau)) 
    = - \sD_{x_0}\phi (t,\tau,x_0, u) f(\tau,x_0,u(\tau)).
\end{align}
Taking the integration with respect to~$\tau$ in~$[t_0,t]$ and using the fundamental theorem of calculus give
\begin{align}\label{pf1:fc}
    \phi (t,t_0,x_0, u) - x_0
    = - \int_{t_0}^t \frac{\partial \phi^u(t,\tau,x_0)}{\partial \tau} \, d\tau
    = \int_{t_0}^t \sD_{x_0} \phi (t,\tau,x_0, u) f(\tau,x_0,u(\tau))  \, d\tau .
\end{align}
From~\eqref{eq:UISS} with~$\delta u(t) \equiv \0_m$, we have, for all~$t_1 \in [t_0, t]$
\begin{align}\label{pf2:fc}
    \| \phi (t,t_1,x_0, u) - x_0 \|_\cX
    &\overset{\eqref{pf1:fc}}{\le} \int_{t_1}^t \| \sD_{x_0} \phi (t,\tau,x_0, u) \|_\cX \, \| f(\tau,x_0,u(\tau)) \|_\cX \, d\tau \nonumber\\
    &\overset{\eqref{eq:UISS}}{\le} k \int_{t_1}^t \e^{b (t -\tau)} \| f(\tau,x_0,u(\tau)) \|_\cX \, d\tau \nonumber\\
    &\; \le k \int_{t_0}^t \e^{b (t -\tau)} \| f(\tau,x_0,u(\tau)) \|_\cX \, d\tau =: R(t_0, t).
\end{align}
Since~$f(\tau, x_0, u(\tau))$ is a continuous function of~$\tau \in \bR$, we have~$R(t_0, t) \in \bR_{\ge 0}$ for all~$(t_0, t) \in \bR \times \bR$ such that~$t \ge t_0$.

Now, we prove the statement by contradiction. Suppose~$T^+ = T^+(t_0, x_0, u) < \infty$. Define~$\rho (t_0):=\sup_{t\in[t_0,T^{+}]}R(t_0, t)<\infty$, and compact set~$B_\cX(x_0, \rho) := \{x \in \bR^n: \|x -x_0 \|_\cX \le \rho \}$. From~\eqref{pf2:fc}, $\phi (t,t_1,x_0, u) \in B_\cX(x_0, \rho (t_0))$ for every $t\in[t_0,T^+)$ and every $t_1 \in[t_0,t]$. Thus, the family of trajectories remains in the compact set~$K:=[t_0,T^+]\times B_\cX(x_0,\rho (t_0))$. By the Picard--Lindel\"of theorem, there exists~$s = s(K) \in \bR_{> 0}$ such that every solution starting from~$K$ extends to~$s$. Applying this to~$\phi(t,t_0,x_0,u)$ at time~$t<T^+$ with~$T^+ - t < s$ continues the solution beyond $T^+$, contradicting the maximality of~$T^+$.
\qed


\subsection{Proof of Theorem~\ref{thm:ISS}}\label{pf:ISS}

The proof relies on the following auxiliary lemma.
\begin{lem}\label{lem:dini}
    For lower semicontinuous function~$w:[0,T] \to \bR$, if $\limsup_{s \to 0^+} \frac{w(t+s)-w(t)}{s} \le 0$ for all~$t \in [0,T]$, then~$w(t_1) \le w(t_0)$ for all~$t_1 \in [t_0,T]$ and~$t_0 \in [0, T]$.
\end{lem}
\begin{pf}
    We prove the statement by contradiction. Suppose that there exist~$t_1 \in (t_0,T]$ and~$t_0 \in [0, T)$ such that~$w(t_0) < w(t_1)$. Define
    \begin{align*}
        \varepsilon := \frac{w(t_1)-w(t_0)}{2(t_1-t_0)}>0,   
    \end{align*}
    and~$\xi(t):=w(t)-w(t_0) - \varepsilon (t-t_0)$. Then,~$\xi(t)$ is lower semicontinuous on~$[0,T]$, and satisfies~$\xi(t_0) = 0$, $\xi(t_1) > 0$, and
    \begin{align}\label{pf:dini}
       \limsup_{s \to 0^+} \frac{\xi(t+s)-\xi(t)}{s} \le - \varepsilon.
    \end{align}
    Next, define~$S := \{t \in [t_0,t_1]: \xi(t) \le 0\}$. By lower semicontinuity of~$\xi(t)$, $S$ is closed. Also,~$t_0 \in S$ and~$t_1 \notin S$. Define~$\tau := \sup \{t \in [t_0, t_1] : t \in S\}$. Then,~$\xi(\tau) \le 0$ and~$\xi(t) > 0$ for all~$t \in (\tau, t_1]$. Consider two cases: 
    \begin{itemize}
        \item if~$\xi(\tau) = 0$, then
        \begin{align*}
            \frac{\xi(\tau + s) - \xi(\tau)}{s} > 0
        \end{align*}
        for all~$s\in(0,t_1-\tau]$;

        \item if~$\xi(\tau) < 0$, then
        \begin{align*}
            \limsup_{s \to 0^+} \frac{\xi(\tau + s) - \xi(\tau)}{s} > \limsup_{s \to 0^+} \frac{- \xi(\tau)}{s} = \infty.
        \end{align*}
    \end{itemize}
    Both cases contradict~\eqref{pf:dini}. By contradiction,~$w(t_1) \le w(t_0)$ for all~$t_0, t_1 \in [0,T]$ such that~$t_1 \ge t_0$.
    \qed
\end{pf}

Now, we are ready to provide the proof of Theorem~\ref{thm:ISS}.

     (item~\ref{i1:ISS} + $\|\sD_u f\|_{\ux} \le \ell$ $\implies$ item~\ref{i3:ISS})
     Define~$w(t) = \e^{-b t} V(t)$. Compute
    \begin{align*}
        \limsup_{s \to 0^+} \frac{w(t+s) - w(t)}{s}
        &\; = \limsup_{s \to 0^+} \frac{\e^{-b (t+s)} V(t+s) - \e^{-b t} V(t)}{s} \\
        &\; \le \limsup_{s \to 0^+} \frac{\e^{-b (t+s)} V(t+s) - \e^{-b t} V(t+s)}{s} 
        + \limsup_{s \to 0^+} \frac{\e^{-b t} V(t+s) - \e^{-b t} V(t)}{s} \\
        &\overset{\eqref{eq:ICC_dini}}{\le} -b \e^{-b t} V(t) + b \e^{-b t} V(t) = 0.
    \end{align*}
    From Lemma~\ref{lem:dini}, we have~$\e^{-b t} V(t) \le \e^{-b t_0} V(t_0)$, i.e.,
    \begin{align}\label{pf0:ISS}
        V(t) \le \e^{b (t - t_0)} V(t_0), \quad \forall t \ge t_0.
    \end{align}
    The bound~\eqref{eq:ICC_bd} yields~\eqref{eq:exbd}.
    From Lemma~\ref{lem:fc}, $T^+(t_0, x_0, u) = \infty$, i.e., there is no finite escape time for each~$(t_0, x_0) \in \bR \times \bR^n$ and every continuous~$u:\bR \to \bR^m$.
     
     Taking the norm of~\eqref{eq:vsys_sol} and applying~\eqref{eq:exbd} and $\|\sD_u f\|_{\ux} \le \ell$ lead to
     \begin{align*}
	 \| \delta x(t) \|_\cX
    &\le \|\sD_{x_0} \phi (t, t_0, x_0, u) \delta x_0\|_\cX
	    + \int_{t_0}^t \bigl( \| \sD_{x_\tau} \phi (t, \tau, x_\tau, u)
    \sD_u f(\tau, x_\tau, u(\tau)) \delta u (\tau)\|_\cX \bigr) \bigr|_{x_\tau = \phi (\tau, t_0, x_0, u)} \; d\tau \nonumber\\
    &\le \|\sD_{x_0} \phi (t, t_0, x_0, u) \delta x_0\|_\cX
	    + \int_{t_0}^t \bigl( \| \sD_{x_\tau} \phi (t, \tau, x_\tau, u)\|_\cX \,
    \|\sD_u f(\tau, x_\tau, u(\tau))\|_{\ux}\bigr) \bigr|_{x_\tau = \phi (\tau, t_0, x_0, u)} \, \|\delta u (\tau)\|_\cU  \; d\tau \nonumber\\
    &\overset{\eqref{eq:exbd}, \|\sD_u f\|_{\ux} \le \ell}{\le}
    k \e^{b (t - t_0)} \| \delta x_0 \|_\cX
        + k \ell \int_{t_0}^t \e^{b (t - \tau)} \|\delta u(\tau) \|_\cU \, d\tau,
    \end{align*}
    i.e.,~\eqref{eq:UISS}.


       (item~\ref{i3:ISS} $\implies$ item~\ref{i4:ISS})
       From Lemma~\ref{lem:fc}, $T^+(t_0, x_0, u) = \infty$.
        Let~$\gamma (s) = s y_0 + (1-s) x_0$ and~$\nu (s) := s v + (1-s) u$.  By the fundamental theorem of calculus and the triangle inequality,
    \begin{align}\label{pf2:ISS}
        \| \phi (t, t_0, y_0, v) - \phi (t, t_0, x_0, u) \|_\cX 
         = \left\| \int_0^1 \frac{\partial \phi (t, t_0, \gamma(s), \nu(s))}{\partial s}  \, ds \right\|_\cX
         \le \int_0^1 \left\| \frac{\partial \phi (t, t_0, \gamma(s), \nu(s))}{\partial s}  \right\|_\cX ds.
    \end{align}
    By Theorem~\ref{thm:GD}, $\frac{\partial\phi(t,t_0,\gamma(s),\nu(s))}{\partial s}$ equals~$\delta x(t)$ along~$\phi(t,t_0,\gamma(s),\nu(s))$ with~$\delta x_0 = \frac{d\gamma (s)}{ds} = y_0 - x_0$
    and~$\delta u = \frac{d\nu (s)}{ds} = v - u$. Applying item~\ref{i3:ISS} gives
    \begin{align}\label{pf3:ISS}
        \left\| \frac{\partial \phi (t, t_0, \gamma(s), \nu(s))}{\partial s} \right\|_\cX
        \le k \e^{b (t - t_0)} \| y_0 - x_0 \|_\cX
            + k \ell \int_{t_0}^t \e^{b (t - \tau)} \| v(\tau) - u(\tau) \|_\cU \, d\tau.
    \end{align}
    Since the right-hand side is independent of~$s \in [0,1]$, substituting~\eqref{pf3:ISS} into~\eqref{pf2:ISS} yields~\eqref{eq:IISS}.


     (item~\ref{i4:ISS} $\implies$ item~\ref{i3:ISS})
     Fix $t\in[t_0,T^+(t_0,x_0,u))$ and $\delta x_0\in\bR^n$. There exists~$h_0>0$ such that the perturbed flow $\phi(t,t_0,x_0+h\delta x_0,u)$ exists for each~$h\in(0,h_0]$ on~$[t_0,t]$. From the continuity of norm~$\|\cdot\|_\cX$ and Theorem~\ref{thm:GD}, we have
    \begin{align*}
     \|\sD_{x_0}\phi(t,t_0,x_0,u) \delta x_0\|_\cX
     = \lim_{h\to0^+}
        \frac{\|\phi(t,t_0,x_0+h\delta x_0,u)-\phi(t,t_0,x_0,u)\|_\cX}{h}.
    \end{align*}
    Substituting $y_0=x_0+h\delta x_0$ and $v\equiv u$ into \eqref{eq:IISS} and dividing by~$h$ give
    \begin{align*}
       \|\sD_{x_0}\phi(t,t_0,x_0,u) \delta x_0\|_\cX \le k \e^{b(t-t_0)}\,\|\delta x_0\|_\cX.
    \end{align*}
    Since~$\delta x_0\in\bR^n$ and~$t\in[t_0,T^+(t_0,x_0,u))$ are arbitrary, we obtain~\eqref{eq:exbd} for all~$t\in[t_0,T^+(t_0,x_0,u))$.
    From Lemma~\ref{lem:fc}, $T^+(t_0, x_0, u) = \infty$.
     
    From~\eqref{eq:sol_GD} and the continuity of the norm~$\| \cdot \|_\cX$, it follows that
        \begin{align*}
            \|\delta x(t)\|_\cX
            = \lim_{h \to 0^+}
                \frac{\|\phi (t, t_0, x_0 + h \delta x_0, u + h \delta u) - \phi (t, t_0, x_0, u)\|_\cX}{h}.
        \end{align*}
    	Next, substituting~$y_0 = x_0 + h \delta x_0$ and~$v = u + h \delta u$ into~\eqref{eq:IISS}, dividing by~$h$, and taking the limit leads to
    	\begin{align*}
            \lim_{h \to 0^+}
                \frac{\| \phi (t, t_0, x_0 + h \delta x_0, u + h \delta u) - \phi (t, t_0, x_0, u) \|_\cX}{h}
        	\le k \e^{b (t - t_0)} \| \delta x_0 \|_\cX
            	+ k \ell \int_{t_0}^t \e^{b (t - \tau)} \| \delta u(\tau) \|_\cU \, d\tau.
	    \end{align*}    	    
        Combining these two yields~\eqref{eq:UISS}.

    (item~\ref{i3:ISS} $\implies$ item~\ref{i1:ISS} + $\|\sD_u f\|_{\ux} \le k \ell$)
    From Lemma~\ref{lem:fc}, $T^+(t_0, x_0, u) = \infty$. Define
    \begin{align}\label{pf4:ISS}
       V(t_0, x_0, \delta x_0)
       := \sup_{\substack{h \ge 0 \\ u:[t_0,t_0+h]\to \bR^m}} \e^{-b h} \| \sD_{x_0} \phi (t_0+h, t_0, x_0, u) \delta x_0 \|_\cX.
    \end{align}
    Since this is a supremum of a jointly continuous function, this is jointly lower semicontinuous in~$(t_0, x_0, \delta x_0)$.
    Taking~$h=0$ gives~$\| \delta x_0 \|_\cX \le V(t_0, x_0, \delta x_0)$. From~\eqref{eq:vsys_sol} and item~\ref{i3:ISS} with~$\delta u \equiv \0_m$, we have~$V(t_0, x_0, \delta x_0) \le  k \| \delta x_0\|_\cX$. Combining these two yields~\eqref{eq:ICC_bd}. 

    Next, we show~\eqref{eq:ICC_dini}. Compute
    \begin{align*}
        &V(t_0+s, \phi(t_0+s,t_0,x_0,u), \sD_{x_0} \phi (t_0+s, t_0, x_0, u) \delta x_0) \nonumber\\
       &\quad \; = \sup_{\substack{h \ge 0 \\ u:[t_0+s,t_0+s+h]\to \bR^m}} \e^{-b h} \| \sD_x \phi (t_0+s+h, t_0+s,\phi(t_0+s,t_0,x_0, u), u) \delta x(t_0+s) \|_\cX \nonumber\\
        &\quad \overset{\eqref{eq3:tm}}{=} \sup_{\substack{h \ge 0 \\ u:[t_0+s,t_0+s+h]\to \bR^m}} \e^{-b h} \| \sD_{x_0} \phi (t_0+s+h, t_0, x_0, u) \delta x_0 \|_\cX.
    \end{align*}
    Introducing~$\tau:= s + h \ge s$, we have
    \begin{align*}
        &V(t_0+s, \phi(t_0+s,t_0,x_0,u), \sD_{x_0} \phi (t_0+s, t_0, x_0, u) \delta x_0) \nonumber\\
        &\quad = \e^{b s} \sup_{\substack{\tau \ge s \\ u:[t_0+s,t_0+\tau]\to \bR^m}} \e^{-b \tau} \| \sD_{x_0} \phi (t_0+\tau, t_0, x_0, u) \delta x_0 \|_\cX \nonumber\\
       &\quad \le \e^{b s} \sup_{\substack{\tau \ge 0 \\ u:[t_0,t_0+\tau]\to \bR^m}} \e^{-b \tau} \| \sD_{x_0} \phi (t_0+\tau, t_0, x_0,u) \delta x_0 \|_\cX 
       \overset{\eqref{pf4:ISS}}{=} \e^{b s} V(t_0, x_0, \delta x_0).
    \end{align*}
    Subtracting~$V(t_0, x_0, \delta x_0)$, dividing by~$s$, and taking the limit superior as~$s\to 0^+$, we obtain~\eqref{eq:ICC_dini}.

    Finally, we show~$\|\sD_u f\|_{\ux} \le k \ell$.
    Substituting~$t = t_0 + h$ with~$h>0$,~$\delta x_0 = \0_n$, and~$\delta u(t) \equiv w \in \bR^m$ into~\eqref{eq:vsys_sol} and~\eqref{eq:UISS} yields
    	\begin{align}\label{pf5:ISS}
     		&\Biggl\| \int_{t_0}^{t_0 + h} \bigl( \sD_{x_\tau} \phi (t_0+h, \tau, x_\tau, u) 
            \sD_u f(\tau, x_\tau, u(\tau)) \bigr)\bigr|_{x_\tau = \phi (\tau, t_0, x_0, u)} \, d \tau \, w  \Biggr\|_\cX
            \le k \ell \| w \|_\cU \int_{t_0}^{t_0 + h} \e^{b (t_0 + h - \tau)} d\tau.
        \end{align}
        Dividing the term inside the norm of the left-hand side by~$h$ and taking the limit superior leads to
        \begin{align}\label{pf6:ISS}
              &\lim_{h \to 0^+} \frac{1}{h} \int_{t_0}^{t_0 + h} \bigl(\sD_{x_\tau} \phi (t_0+h, \tau, x_\tau, u)
             \sD_u f(\tau, x_\tau, u(\tau))\bigr)\bigr|_{x_\tau = \phi (\tau, t_0, x_0, u)} d \tau
            \nonumber\\
            &\quad =  \sD_{x_0} \phi (t_0, t_0, x_0, u) \sD_u f(t_0, x_0, u(t_0))
            \overset{\eqref{eq2:tm}}{=}
                \sD_u f(t_0, x_0, u(t_0)).
        \end{align}
        Also, dividing the right-hand side of~\eqref{pf5:ISS} by~$h$ and taking the limit superior leads to
        \begin{align}\label{pf7:ISS}
            k \ell \| w \|_\cU \limsup_{h \to 0^+} \frac{1}{h} \int_{t_0}^{t_0 + h} \e^{b (t_0 + h - \tau)} d\tau
            = k \ell \| w \|_\cU.
    	\end{align}
        The continuity of the norm~$\| \cdot \|_\cU$ and~\eqref{pf5:ISS}--\eqref{pf7:ISS} imply
        \begin{align*}
            \| \sD_u f(t_0, x_0, u(t_0)) w \|_\cX \le k \ell \| w \|_\cU .
    	\end{align*}
        Since~$w \in \bR^m$ is arbitrary, we have~$\|\sD_u f\|_{\ux} \le k \ell$. 

       (item~\ref{i1:ISS} $\implies$ item~\ref{i2:ISS})
        From the above proofs, one notices that item~\ref{i1:ISS} is equivalent to
        \begin{align}\label{pf8:ISS}
                \| \phi  (t, t_0,  y_0, u) - \phi  (t, t_0, x_0, u) \|_\cX
           		\le k \e^{b (t - t_0)} \| y_0 - x_0 \|_\cX
       		\end{align}
        	for all~$(t, t_0) \in \bR_{\ge t_0} \times \bR$, $(x_0, y_0) \in \bR^n \times \bR^n$, and continuous~$u: \bR \to \bR^m$. We show that item~\ref{i2:ISS} is also equivalent to~\eqref{pf8:ISS}.

        By the same argument that~\eqref{eq:ICC_dini} yields~\eqref{pf0:ISS}, \eqref{eq:ILC_dini} gives~$W(t) \le \e^{b (t - t_0)} W(t_0)$. Thus, \eqref{eq:ILC_bd} implies~\eqref{pf8:ISS}. Conversely, define
    \begin{align}\label{pf9:ISS}
       W(t_0, x_0, y_0)
       := \sup_{\substack{h \ge 0 \\ u:[t_0,t_0+h]\to \bR^m}} \e^{-b h} \|  \phi (t_0+h, t_0, x_0, u) - \phi (t_0+h, t_0, y_0, u) \|_\cX.
    \end{align}
    Then, one can confirm that~$W(t_0, x_0, \delta y_0)$ satisfies~\eqref{eq:ILC} by the same reasoning that~$V(t_0, x_0, \delta x_0)$ in~\eqref{pf4:ISS} satisfies~\eqref{eq:ICC}.
    \qed


\subsection{Proof of Theorem~\ref{dthm:ISS}}
     (item~\ref{di1:ISS} + $\|\sD_u f\|_{\ux} \le \ell$ $\implies$ item~\ref{di3:ISS})
     Recursively applying~\eqref{deq:ICC_dini} and utilizing the bound~\eqref{deq:ICC_bd} yield
    \begin{align}\label{dpf0:ISS}
        \| \delta x(t) \|_\cX \le k b^{t-t_0} \| \delta x_0\|_\cX
    \end{align}
    for~$\delta u = \0_m$.

    Next, taking the norm of~\eqref{deq:vsys} yields
    \begin{align}\label{dpf1:ISS}
        \| \delta x(t+1) \|_\cX 
        &= \bigl( \|\sD_x f(t, x_t, u(t)) \delta x(t) + \sD_u f(t, x_t, u(t)) \delta u (t) \|_\cX \bigr)\bigr|_{x_t=\phi (t, t_0, x_0, u)}  \nonumber\\
        &\le \bigl(\|\sD_x f(t, x_t, u(t))\|_\cX \, \|\delta x(t)\|_\cX + \|\sD_u f(t, x_t, u(t))\|_\ux \, \|\delta u (t)\|_\cU\bigr)\bigr|_{x_t=\phi (t, t_0, x_0, u)} \nonumber\\
        &\overset{\|\sD_u f\|_{\ux} \le \ell}{\le}
        \bigl(\|\sD_x f(t, x_t, u(t))\|_\cX \, \|\delta x(t)\|_\cX\bigr)\bigr|_{x_t=\phi (t, t_0, x_0, u)} + \ell \|\delta u (t)\|_\cU.
    \end{align}
    Combining~\eqref{dpf0:ISS} and~\eqref{dpf1:ISS} concludes~\eqref{deq:UISS}.


       (item~\ref{di3:ISS} $\iff$ item~\ref{di4:ISS})
       This equivalence can be shown similarly to the proof of Theorem~\ref{thm:ISS}.


    (item~\ref{di3:ISS} $\implies$ item~\ref{di1:ISS} + $\|\sD_u f\|_{\ux} \le k \ell$)
    For~$(t, x, \delta x) \in \bZ \times \bR^n \times \bR^n$, define
    \begin{align}\label{dpf4:ISS}
       V(t, x, \delta x)
       := \sup_{\substack{h \in \bZ_{\ge 0} \\ u:[t,t+h] \cap \bZ \to \bR^m}} \frac{ \| \sD_x \phi (t+h, t, x, u) \delta x \|_\cX}{b^h}.
    \end{align}
    Since this is a supremum of a jointly continuous function, this is jointly lower semicontinuous in~$(x, \delta x)$.
    Taking~$h=0$ gives~$\| \delta x \|_\cX \le V(t, x, \delta x)$. From~\eqref{deq:sol_GD} and item~\ref{di3:ISS} with~$\delta u \equiv \0_m$, we have~$V(t, x, \delta x) \le  k \| \delta x\|_\cX$. Combining these two yields~\eqref{deq:ICC_bd}. 

    Next, we show~\eqref{deq:ICC_dini}. Compute
    \begin{align*}
        V\left(t+1, f(t, x, u), \sD_x f(t, x, u) \delta x + \sD_u f(t, x, u) \delta u \right)
       = \sup_{\substack{h \in \bZ_{\ge 0} \\ u:[t+1,t+h+1] \cap \bZ \to \bR^m}} \frac{ \| \sD_x \phi (t+h+1, t, x, u) \delta x \|_\cX}{b^h}.
    \end{align*}
    Introducing~$\tau:= h + 1 \ge 1$, we have
    \begin{align*}
        V\left(t+1, f(t, x, u), \sD_x f(t, x, u) \delta x + \sD_u f(t, x, u) \delta u \right) 
        & = b \sup_{\substack{\tau \in \bZ_{\ge 1} \\ u:[t+1,\tau] \cap \bZ \to \bR^m}} \frac{ \| \sD_x \phi (t+\tau, t, x, u) \delta x \|_\cX}{b^\tau} \nonumber\\
       &\le b \sup_{\substack{\tau \in \bZ_{\ge 0} \\ u:[t,\tau] \cap \bZ \to \bR^m}} \frac{ \| \sD_x \phi (t+\tau, t, x, u) \delta x \|_\cX}{b^\tau}
       \overset{\eqref{dpf4:ISS}}{=} b V(t, x, \delta x).
    \end{align*}
    Thus, we obtain~\eqref{deq:ICC_dini}.

    Finally, $\|\sD_u f\|_{\ux} \le k \ell$ is derived by
    substituting~$t = t_0 + 1$ and~$\delta x_0 = \0_n$ into~\eqref{deq:UISS}.
    

       (item~\ref{di1:ISS} $\implies$ item~\ref{di2:ISS})
        This equivalence can be shown similarly to the proof of Theorem~\ref{thm:ISS} by using
    \begin{align}\label{dpf9:ISS}
       W(t, x, y)
       := \sup_{\substack{h \in \bZ_{\ge 0} \\ u:[t,t+h] \cap \bZ \to \bR^m}} \frac{ \| \phi (t+h, t, x, u) - \phi (t+h, t, y, u) \|_\cX}{b^h}.
    \end{align}
    This completes the proof.    
    \qed    


\begin{ack}                               
This work of Y. Kawano is supported in part by JST FOREST Program Grant Number JPMJFR222E and JSPS KAKENHI Grant Number JP24K00910. This work of F. Bullo is supported in part by AFOSR grant FA9550-22-1-0059. 
\end{ack}


\bibliographystyle{plain}        
\bibliography{autosam-extended-discrete}           

\end{document}